\documentclass[natbib,9pt]{sigplanconf}

\usepackage{amsthm,amsmath,amssymb,latexsym}
\usepackage{upgreek}
\usepackage{graphicx}

\usepackage{code}
\usepackage{comment}
\usepackage{multirow}
\usepackage[all]{xy}

\newcommand{\ie}{\emph{i.e.}}
\newcommand{\eg}{\emph{e.g.}}

\newcommand{\syn}[1]{\mathsf{#1}}
\newcommand{\var}[1]{\mathit{#1}}
\newcommand{\s}[1]{\mathit{#1}}

\newcommand{\parto}{\rightharpoonup}

\newcommand{\set}[1]{\left\{#1\right\}}
\newcommand{\setbuild}[2]{\left\{ #1 : #2\right\}}

\newcommand{\Pow}[1]{{\mathcal{P}\left(#1\right)}}

\newcommand{\Nats}{{\mathbb{N}}}

\newcommand{\vect}[1]{\langle #1\rangle}

\newcommand{\To}{\mathrel{\Rightarrow}}

\newcommand{\wt}{\sqsubseteq}

\newcommand{\join}{\sqcup}

\newcommand{\bigjoin}{\bigsqcup}

\newcommand{\sembr}[1]{\ensuremath{[\![{#1}]\!]}}

\newcommand{\opor}{\mathrel{|}}

\newcommand{\produces}{\mathrel{::=}}

\newcommand{\vv}{v}

\newcommand{\lam}{\ensuremath{\var{lam}}}

\newcommand{\lamterm}{$\lambda$-term}
\newcommand{\lc}{$\lambda$-calculus}

\newcommand{\call}{\ensuremath{\var{call}}}

\newcommand{\free}{\mathit{free}}

\newcommand{\ttfs}{\mbox{\tt .}}

\newcommand{\ttlp}{\mbox{\tt (}}
\newcommand{\ttrp}{\mbox{\tt )}}
\newcommand{\ttlc}{\mbox{\tt \{}}
\newcommand{\ttrc}{\mbox{\tt \}}}
\newcommand{\appform}[2]{\ttlp #1\; #2\ttrp}
\newcommand{\lamform}[2]{\ttlp \uplambda\;\ttlp#1\ttrp\;#2\ttrp}

\newcommand{\ttsc}{\mbox{\tt ;}}

\newcommand{\stmt}{s}

\newcommand{\lab}{\ell}

\newcommand{\expr}{e}

\newcommand{\Eval}{{\mathcal{E}}}

\newcommand{\State}{\Sigma}
\newcommand{\state}{\varsigma}

\newcommand{\Den}{D} 
\newcommand{\Obj}{Obj} 

\newcommand{\tf}{f}

\newcommand{\den}{d}
\newcommand{\obj}{o}

\newcommand{\store}{\sigma}

\newcommand{\env}{\rho}
\newcommand{\benv}{\beta}
\newcommand{\clo}{\var{clo}}
\newcommand{\cont}{\kappa}
\newcommand{\contptr}{\ptr{\cont}}

\newcommand{\ptr}[1]{{p\!\!^{^{#1}}}}
\newcommand{\alloc}{\mathit{alloc}}

\newcommand{\addr}{a}

\newcommand{\val}{\var{val}}

\newcommand{\tm}{t}
\newcommand{\tick}{{{tick}}}

\newcommand{\aTo}{\leadsto}

\newcommand{\sa}[1]{\widehat{\mathit{#1}}}

\newcommand{\aEval}{{\hat{\mathcal{E}}}}

\newcommand{\atf}{{\hat{f}}}

\newcommand{\aState}{{\hat{\Sigma}}}
\newcommand{\astate}{{\hat{\varsigma}}}

\newcommand{\aDen}{\hat{D}} 

\newcommand{\astore}{{\hat{\sigma}}}
\newcommand{\aenv}{{\hat{\rho}}}
\newcommand{\abenv}{{\hat{\beta}}}
\newcommand{\aclo}{{\widehat{\var{clo}}}}
\newcommand{\aobj}{{\hat \obj}}

\newcommand{\acont}{{\hat{\kappa}}}
\newcommand{\acontptr}{\ptr{\acont}}

\newcommand{\aden}{{\hat{d}}}

\newcommand{\aaddr}{{\hat{\addr}}}
\newcommand{\aalloc}{{\widehat{alloc}}}

\newcommand{\aval}{{\widehat{\var{val}}}}

\newcommand{\atick}{{\widehat{tick}}}
\newcommand{\atm}{{\hat t}}

\newcommand{\absmap}{\alpha}
\newcommand{\abs}[1]{|#1|}

\newcommand{\sstate}{\xi}
\newcommand{\SState}{\Xi}

\newcommand{\new}{{\mathit{new}}}
\newcommand{\anew}{{\widehat{\new}}}

\newcommand{\asstate}{\hat \xi}
\newcommand{\aSState}{\hat \Xi}

\newcommand{\nCFA}{$m$-CFA}
\newcommand{\kCFA}{$k$-CFA}

\newcommand{\setbuildsm}[2]{\{ #1 \mathrel{:} #2 \}}

\newtheorem{theorem}{Theorem}[section]
\newtheorem{lemma}[theorem]{Lemma}

\newcommand{\methodDef}{M}
\newcommand{\constDef}{K}

\newcommand{\classform}[3]{
  \mathtt{class}\; 
  #1\;
  \mathtt{extends}\;
  #2\;
  \ttlc
  #3
  \ttrc
}
\newcommand{\con}

\newcommand{\className}{C}
\newcommand{\fieldName}{f}
\newcommand{\methodName}{m}

\newcommand{\MethodLookup}{{\mathcal{M}}}
\newcommand{\FetchRuctor}{{\mathcal{C}}}
\newcommand{\Ructor}{{\mathcal{K}}}

\newcommand{\aMethodLookup}{{\hat {\mathcal{M}}}}
\newcommand{\aFetchRuctor}{{\hat {\mathcal{C}}}}
\newcommand{\aRuctor}{{\hat {\mathcal{K}}}}

\newcommand{\ssucc}{\mathit{succ}}

\begin{document}

\conferenceinfo{PLDI'10,} {June 5--10, 2010, Toronto, Ontario, Canada.}
\CopyrightYear{2010}
\copyrightdata{978-1-4503-0019/10/06}


\title{Resolving and Exploiting the \(\boldsymbol k\)-CFA Paradox}
\subtitle{Illuminating Functional vs. Object-Oriented Program Analysis}

\authorinfo{Matthew Might}{University of Utah}{might@cs.utah.edu}
\authorinfo{Yannis Smaragdakis}{University of Massachusetts}{yannis@cs.umass.edu}
\authorinfo{David Van Horn}%
{Northeastern University}{dvanhorn@ccs.neu.edu}

\maketitle

\begin{abstract}
  Low-level program analysis is a fundamental problem, taking the shape
of ``flow analysis'' in functional languages and ``points-to''
analysis in imperative and object-oriented languages.  Despite the
similarities, the vocabulary and results in the two communities remain
largely distinct, with limited cross-understanding.  One of the few
links is Shivers's \kCFA{} work, which has advanced the concept of
``context-sensitive analysis'' and is widely known in both
communities.

Recent results indicate that the relationship between the functional
and object-oriented incarnations of $k$-CFA is not as well understood
as thought.  Van Horn and Mairson proved \kCFA{} for $k \geq 1$ to be
EXPTIME-complete; hence, no polynomial-time algorithm can exist.  Yet,
there are several polynomial-time formulations of context-sensitive
points-to analyses in object-oriented languages.  Thus, it seems that
functional \kCFA{} may actually be a profoundly different analysis
from object-oriented \kCFA{}.  We resolve this paradox by showing that
the exact same specification of \kCFA{} is polynomial-time for
object-oriented languages yet exponential-time for functional ones:
objects and closures are subtly different, in a way that interacts
crucially with context-sensitivity and complexity.  This illumination
leads to an immediate payoff: by projecting the object-oriented
treatment of objects onto closures, we derive a polynomial-time
hierarchy of context-sensitive CFAs for functional programs.

\end{abstract}







\section{Introduction}

One of the most fundamental problems in program analysis is
determining the entities to which an expression may refer at
run-time.  In imperative and object-oriented (OO) languages, this is
commonly phrased as a \emph{points-to} (or \emph{pointer}) analysis:
to which objects can a variable point? In functional languages, the
problem is called \emph{flow analysis} \cite{dvanhorn:Midtgaard2011Controlflow}:
to which expressions can a value flow?

Both points-to and flow analysis acquire a degree of complexity for
higher-order languages: functional languages have
first-class functions and object-oriented languages have dynamic
dispatch; these features conspire to make call-target resolution
depend on the flow of values, even as the flow of values depends on
what targets are possible for a call.
That is, data-flow depends on control-flow, yet control-flow depends on
data-flow.
Appropriately, this problem is commonly called
\emph{control-flow analysis} (CFA).

Shivers's \kCFA{} \cite{mattmight:Shivers:1991:CFA} is a
well-known family of control-flow analysis algorithms, widely
recognized in both the functional and the object-oriented
world. \kCFA{} popularized the  idea of context-sensitive
flow analysis.\footnote{Although the \kCFA{} work is often used as a
  synonym for ``$k$-context-sensitive'' in the OO world, \kCFA{} is more
  correctly an algorithm that packages context-sensitivity together
  with several other design decisions. In the terminology of OO
  points-to analysis, \kCFA{} is a $k$-call-site-sensitive,
  field-sensitive points-to analysis algorithm with a
  context-sensitive heap and with on-the-fly call-graph
  construction. (\citet{Lhotak:2006:PAU} and \citet{1391987} are good
  references for the classification of points-to analysis algorithms.)
  In this paper we use the term ``\kCFA'' with this more precise meaning,
  as is common in the functional programming world, and not just as a
  synonym for ``$k$-context-sensitive''. Although this classification 
  is more precise, it still allows for a range of algorithms, as we
  discuss later.} Nevertheless, there have
always been annoying discrepancies between the experiences in the
application of \kCFA{} in the functional and the OO world. Shivers
himself notes in his ``Best of PLDI'' retrospective that ``the basic
analysis, for any $k > 0$ [is] intractably slow for large programs''
\cite{dvanhorn:shivers-sigplan04}.
This contradicts common experience in the OO setting, where a 1- and
2-CFA analysis is considered heavy but certainly possible
\cite{1391987,BS-OOPSLA09}. 

To make matters formally worse,
\citet{dvanhorn:VanHorn-Mairson:ICFP08} recently proved \kCFA{} for $k
\geq 1$ to be EXPTIME-complete, i.e., non-polynomial.
%
%
Yet the OO formulations of \kCFA{} have provably polynomial complexity
(e.g., \citet{BS-OOPSLA09} express the algorithm in Datalog, which is
a language that can only express polynomial-time algorithms). This
paradox seems hard to resolve. Is \kCFA{} misunderstood?  Has 
inaccuracy crept into the transition from functional to OO?

In this paper we resolve the paradox and illuminate the deep
differences between functional and OO context-sensitive program
analyses. We show that the exact same formulation of \kCFA{} is
exponential-time for functional programs yet polynomial-time for OO
programs. To ensure fidelity, our proof appeals directly to Shivers's
original definition of \kCFA{} and applies it to the most common
formal model of Java, Featherweight Java.

As might be expected, our finding hinges on the fundamental
difference between typical functional and OO languages: the former
create implicit closures when lambda expressions are created, while
the latter require the programmer to explicitly ``close'' (i.e., pass
to a constructor) the data that a newly created object can reference.
At an intuitive level, this difference also explains why the exact
same \kCFA{} analysis will not yield the same results if a functional
program is automatically rewritten into an OO program: the call-site
context-sensitivity of the analysis leads to loss of precision when
the values are explicitly copied---the analysis merges the information
for all paths with the same $k$-calling-context into the same entry
for the copied data.

Beyond its conceptual significance, our finding pays immediate
dividends: By emulating the behavior of OO \kCFA{}, we
derive a hierarchy, \nCFA{}, of polynomial CFA analyses for functional
programs. In technical terms, \kCFA{} corresponds to an abstract
interpretation over shared-environment closures, while \nCFA{} 
corresponds to an abstract interpretation over flat-environment
closures.  
\nCFA{} turns out to be an important instantiation in the space of
analyses described by
\citet{mattmight:Jagannathan:1995:Unified}.

\section{Background and Illustration}

Although we prove our claims formally in later sections, we first
illustrate the behavior of \kCFA{} for OO and functional
programs informally, so that the reader has an intuitive understanding of the
essence of our argument.

\subsection{Background: What is CFA?}
\kCFA{} was developed to solve the higher-order control-flow problem
in \lc-based programming languages.
%
%
Functional languages are explicitly vulnerable to the higher-order
control-flow problem, because closures are passed around as
first-class values.
Object-oriented languages like Java are implicitly higher-order,
because method invocation is resolved dynamically---the invoked method
depends on the type of the object that makes it to the invocation
point.
%

In practice, CFAs must compute much more than just control-flow
information.  CFAs are also data-flow analyses, computing the values
that flow to any program expression. In the object-oriented setting,
CFA is usually termed a ``points-to'' analysis and the interplay
between control- and data-flow is called ``on-the-fly call-graph
construction'' \cite{Lhotak:2006:PAU}.

Both the functional community and the pointer-analysis community have
assigned a meaning to the term \kCFA.
Informally, \kCFA{} refers to a hierarchy of global static analyses
whose context-sensitivity is a function of the last $k$ call sites 
visited.
In its functional formulation, \kCFA{} uses this context-sensitivity
for every value and variable---thus, in pointer analysis terms,
\kCFA{} is a $k$-call-site-sensitive analysis with a $k$-context-sensitive
heap.

\subsection{Insight and Example}

The paradox prompted by the Van Horn and Mairson proofs seems to imply
that \kCFA{} actually refers to two different analyses: one for
functional programs, and one for object-oriented/imperative programs.
The surprising finding of our work is that \kCFA{} means the same
thing for both programming paradigms, but that its behavior is different
for the object-oriented case.

\kCFA{} was defined by abstract interpretation of the \lc{} semantics for
an abstract domain collapsing data values to static abstractions
qualified by $k$ calling contexts. Functional implementations of the
algorithm are often heavily influenced by this abstract interpretation
approach.  The essence of the exponential complexity of \kCFA{} (for
$k \geq 1$) is that, although each variable can appear with at most
$O(n^{k})$ calling contexts, the number of variable environments is
exponential, because an environment can combine variables from
distinct calling contexts.  Consider the following term:
\begin{displaymath}
\lamform{z}{\appform{z}{x_1 \ldots x_n}}\enspace\text.
\end{displaymath}
This expression has $n$ free variables.  In 1-CFA, each variable is
mapped to the call-site in which it was bound.  By binding each of the
$x_i$ in multiple call-sites, we can induce an exponential number of
environments to close this $\lambda$-term:
\begin{displaymath}
\begin{array}{l}
\ttlp\lamform{f_1}{\appform{f_1}{0} \appform{f_1}{1}}\\
\mbox{\tt{\ }}\ttlp\uplambda\;\ttlp x_1\ttrp\\
\quad\cdots\\
\quad\ttlp\lamform{f_n}{\appform{f_n}{0} \appform{f_n}{1}}\\
\quad\mbox{\tt{\ }}\ttlp\uplambda\;\ttlp x_n\ttrp\; \\
\quad\mbox{\tt{\ \ }}\lamform{z}{\appform{z}{x_1 \ldots x_n}}\ttrp\ttrp\cdots\ttrp\ttrp\enspace\text.
\end{array}
\end{displaymath}
Notice that each $x_i$ is bound to $0$ and $1$, thus there are $2^n$
environments closing the inner $\lambda$-term.

The same behavior is not possible in the object-oriented setting
because creating closures has to be explicit (a fundamental difference
of the two paradigms\footnote{It is, of course, impossible to strictly
classify languages by paradigm (``what is JavaScript?'') so our
statements reflect typical, rather than universal, practice.}) and the
site of closure creation becomes the common calling context for all
closed variables.

\begin{figure*}[tbp]
\begin{center}
\includegraphics*[scale=0.64, viewport=40 130 760 565]{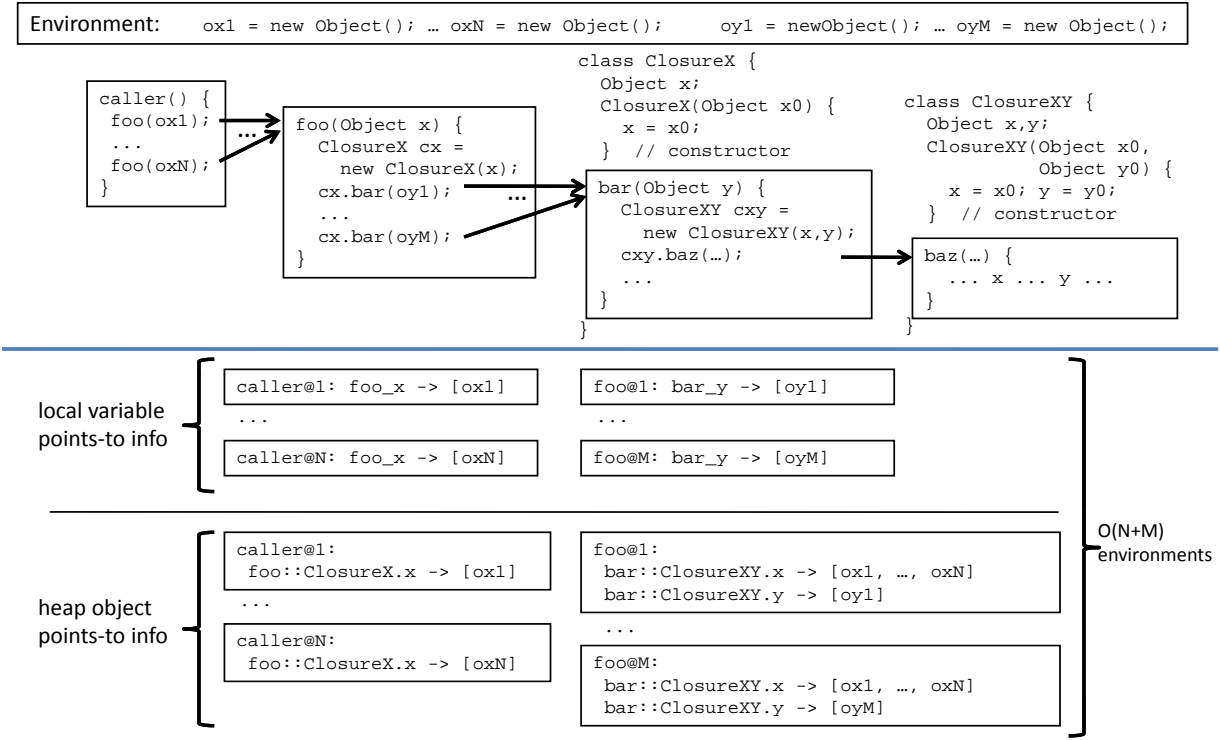}

\caption{An example OO program, analyzed under 1-CFA{}. Parts that are
 orthogonal to the analysis (e.g., return types, the class containing
 \texttt{foo}, the body of \texttt{baz}) are elided. The bottom part
 shows the (points-to) results of the analysis in the form
 ``\emph{context}: \emph{var} \texttt{->} \emph{abstractObject}''.
 Conventions: we use \texttt{[ox1]}, ..., \texttt{[oxN]},
 \texttt{[oy1]}, ..., \texttt{[oyM]} to mean the abstract objects
 pointed to by the corresponding environment variables. (We only care
 that these objects be distinct.)  \emph{method}\texttt{\_}\emph{var}
 names a local variable, \emph{var} inside a
 method. \emph{method}\texttt{::}\emph{Type}\texttt{.}\emph{field}
 refers to a field of the object of type \emph{Type} allocated inside
 \emph{method}.  (This example allocates a single object per method,
 so no numeric distinction of allocation sites is necessary.)
 \emph{callermethod}\texttt{@}\emph{num} designates the \emph{num}-th
 call-site inside method \emph{callermethod}.}
\label{fig:illustration-oo}
\end{center}
\end{figure*}

\begin{figure*}[tbp]
\begin{center}
\includegraphics*[scale=0.64, viewport=40 155 760 536]{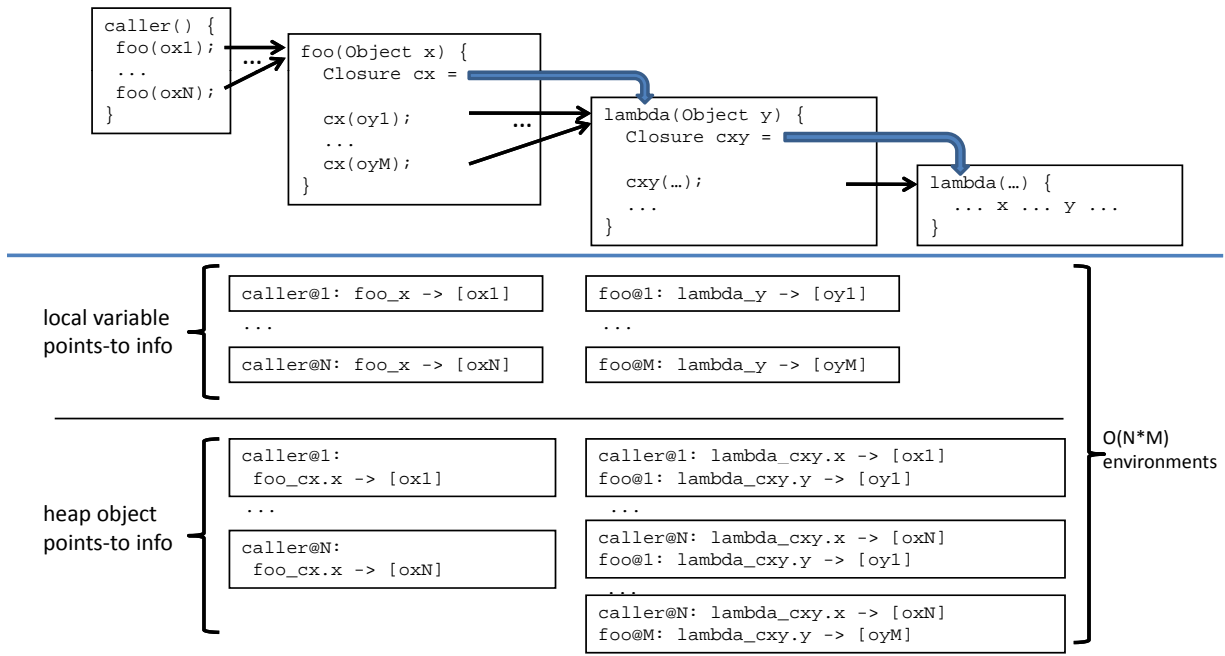}

\caption{The same program in functional form (implicit
closures). \emph{The lambda expressions are drawn outside their
lexical environment to illustrate the analogy with the OO code.}  The
number of environments out of the abstract interpretation is now
$O(NM)$ because variables \texttt{x} and \texttt{y} in the rightmost
lambda were not closed together and have different contexts.}
\label{fig:illustration-fun}
\end{center}
\end{figure*}

Figures~\ref{fig:illustration-oo} and \ref{fig:illustration-fun}
demonstrate this behavior for a 1-CFA analysis. (This is the shortest,
in terms of calling depth, example that can demonstrate the
difference.) Figure~\ref{fig:illustration-oo} presents the program in
OO form, with explicit closures---i.e., objects that are initialized
to capture the variables that need to be used
later. Figure~\ref{fig:illustration-fun} shows the same program in
functional form. We use a fictional (for Java) construct
\texttt{lambda} that creates a closure out of the current environment.
The bottom parts of both figures show the information that the
analysis computes. (We have grouped the information in a way that is
more reflective of OO \kCFA{} implementations, but this is just a
matter of presentation.)

The essential question is ``in how many environments does function
\texttt{baz} get analyzed?'' The exact same,
abstract-interpretation-based, 1-CFA{} algorithm produces $O(N+M)$
environments for the object-oriented program and $O(NM)$ environments
for the functional program. The reason has to do with how the
context-sensitivity of the analysis interacts with the explicit
closure. Since closures are explicit in the OO program, all
(heap-)accessible variables were closed simultaneously. One can see
this in terms of variables \texttt{x} and \texttt{y}: both are closed
by copying their values to the \texttt{x} and \texttt{y} fields of an
object in the expression ``\texttt{new ClosureXY(x,y)}''.  This
copying collapses all the different values for \texttt{x} that have
the same 1-call-site context. Put differently, \texttt{x} and
\texttt{y} inside the OO version of \texttt{baz} are not the original
variables but, rather, copies of them. The act of copying, however,
results in less precision because of the finite context-sensitivity of
the analysis. In contrast, the functional program makes implicit
closures in which the values of \texttt{x} and \texttt{y} are closed at
different times and maintain their original context. The abstract
interpretation results in computing all $O(NM)$ combinations of
environments with different contexts for \texttt{x} and
\texttt{y}. (If the example is extended to more levels, the number
of environments becomes exponential in the length of the program.)

The above observations immediately bring to mind a well-known result
in the compilation of functional languages: the choice between shared
environments and flat environments~\cite[page
  142]{dvanhorn:Appel1991Compiling}.
In a flat environment, the values of all free variables are copied
into a new environment upon allocation.
In a flat-environment scenario, it \emph{is} sufficient to know only
the base address of an environment to look up the value of a variable.
To define the meaning of a program, it clearly makes no difference which
environment representation a formal semantics models.
However, in \emph{compilation} there are trade-offs: shared
environments make closure-creation fast and variable look-up slow,
while flat environments make closure-creation slow and variable
look-up fast.
The choice of environment representation also makes a profound
difference during abstract interpretation.

\section{Shivers's original $\boldsymbol k$-CFA}
\label{sec:kcfa-cps}

Because one possible resolution to the paradox is
that \kCFA{} for object-oriented programs and \kCFA{} for the \lc{} is
just a case of using the same name for two different concepts, we need
to be confident that the analysis we are working with is really
\kCFA.
To achieve that confidence, we return to the source of
\kCFA---Shivers's dissertation~\cite{mattmight:Shivers:1991:CFA},
which formally and precisely pins down its meaning.
We take only cosmetic liberties in reformulating Shivers's \kCFA---we
convert from a tail-recursive denotational semantics to a small-step
operational semantics, and we rename contours to times.
Though equivalent, Shivers's original formulation of \kCFA{} differs
significantly from later ones; readers familiar with only
modern CFA theory may even find it unusual.
Once we have reformulated \kCFA, our goal will be to adapt it as
literally as possible to Featherweight Java.

%

\subsection{A grammar for CPS}

A minimal grammar for CPS (Figure~\ref{fig:cps-grammar}) contains two 
expression forms---\lamterm s and variables---and one call form.
\begin{figure}
\begin{small}
\begin{gather*}
\begin{align*}
  \lam \in \syn{Lam} &\produces \lamform{v_1 \ldots v_n}{\call}^\lab
  &
  \vv \in \syn{Var}
  \\  
  \call \in \syn{Call} &\produces \appform{f}{e_1 \ldots e_n}^\lab
  &
  f,\expr \in \syn{Exp} &= \syn{Var} + \syn{Lam}  
\end{align*}
  \\
  \lab \in \syn{Lab} \text{ is a set of labels}
\end{gather*}\end{small}%
\caption{Grammar for CPS}
\label{fig:cps-grammar}
\end{figure}
The body of every \lamterm{} is a call site, which ensures the CPS
constraint that functions cannot directly return to their callers.
We also attach a unique label to every \lamterm{} and call site.

\subsection{Concrete semantics for CPS}
We model the semantics for CPS as a small-step state machine.
Each state in this machine contains the current call site, a binding
environment in which to evaluate that call, a store and a time-stamp:
\begin{small}\begin{align*}
  \state \in \State &=
  \syn{Call} \times 
  \s{BEnv} \times 
  \s{Store} \times
  \s{Time}
  \\
  \benv \in \s{BEnv} &= \syn{Var} \parto \s{Addr} 
  \\
  \store \in \s{Store} &= \s{Addr} \parto \s{D}
  \\
  \den \in \Den &= \s{Clo}
  \\
  \clo \in \s{Clo} &= \syn{Lam} \times \s{BEnv}
  \\
  \addr \in \s{Addr} &\text{ is an infinite set of addresses}
  \\
  \tm \in \s{Time} &\text{ is an infinite set of time-stamps}
  \text.
\end{align*}\end{small}%
Environments in this state-space are factored; instead of mapping a
variable directly to a value, a binding environment maps a variable to
an address, and then the store maps addresses to values.
The specific structure of both time-stamps and addresses will be
determined later.
Any infinite set will work for either addresses or time-stamps for the
purpose of defining the meaning of the concrete semantics.
(Specific choices for these sets can simplify proofs of soundness,
which is why they are left unfixed for the moment.)

To inject a call site $\call$ into an initial state, we pair it with an empty
environment, an empty store and a distinguished initial time:
\begin{small}\begin{align*}
  \state_0 &= (\call, [], [], \tm_0)
  \text.
\end{align*}\end{small}%

The concrete semantics are composed of an evaluator for expressions
and a transition relation on states:
\begin{small}\begin{align*}
  \Eval &: \syn{Exp} \times \s{BEnv} \times \s{Store} \parto \s{D}
  &
  (\To) &\subseteq \State \times \State
  \text.
\end{align*}\end{small}%
The evaluator looks up variables, and creates closures over \lamterm s:
\begin{small}\begin{align*}
  \Eval(v,\benv,\store) &= \store(\benv(v))
  &
  \Eval(\lam,\benv,\store) &= (\lam,\benv)
  \text.
\end{align*}\end{small}%
In CPS, there is only one rule to transition from one state to another;
when $\call = \sembr{\appform{f}{e_1 \ldots e_n}^\lab}$:
\begin{gather*}
  (\call,\benv,\store,\tm) 
  \To
  (\call',\benv'',\store',\tm')\text{, where }\\
\begin{small}\begin{align*}
  (\lam,\benv') &= \Eval(f,\benv,\store)
  &
  \den_i &= \Eval(e_i,\benv,\store)
  \\
  \lam &= \sembr{\lamform{v_1 \ldots v_n}{\call'}^{\lab'}}
  &
  \tm' &= \tick(\call,\tm)
  \\
  \addr_i &= \alloc(v_i,\tm')
  &
  \benv'' &= \benv'[v_i \mapsto \addr_i]
  \\
  \store' &= \store[\addr_i \mapsto \den_i]
  \text.
\end{align*}\end{small}%
\end{gather*}
There are two external parameters to this semantics, a function for
incrementing the current time-stamp and a function for allocating
fresh addresses for bindings:
\begin{small}\begin{align*}
  \mathit{tick} &: \syn{Call} \times \s{Time} \to \s{Time}
  \\
  \alloc &: \syn{Var} \times \s{Time} \to \s{Addr}
\end{align*}\end{small}%
It is possible to define a semantics in which the $\tick$ function
does not have access to the current call site, but providing access to
the call site will end up simplifying the proof of soundness for
\kCFA.

Naturally, we expect that new time-stamps and addresses are always unique; formally:
\begin{small}\begin{align}
  &\tm < \tick(\call,\tm)\text.
  \\
  &\text{If } 
  v \neq v'
  \text{, then }
  \alloc(v,\tm) \neq
  \alloc(v',\tm)
  \text. 
  \\
  &\text{If } 
  \tm \neq \tm'
  \text{, then }
  \alloc(v,\tm) \neq
  \alloc(v',\tm')
  \text. 
\end{align}\end{small}%
For the sake of understanding the concrete semantics, the obvious
solution to these constraints is to use the natural numbers for time:
\begin{small}\begin{align*}
  \s{Time} &= \Nats
  &
  \s{Addr} &= \syn{Var} \times \s{Time}
  \text,
\end{align*}\end{small}%
so that the $\tick$ function merely has to increment:
\begin{small}\begin{align*}
  \mathit{tick}(\_,\tm) & = \tm + 1
  &
  \alloc(\vv,\tm) & = (\vv,\tm)
  \text.
\end{align*}\end{small}%

\subsection{Executing the concrete semantics}


The concrete semantics finds the set of states reachable from the
initial state.
The system-space for this process is a set of states:
\begin{small}\begin{align*}
  \sstate \in \SState &= \Pow{\State}
  \text.
\end{align*}\end{small}%
The system-space exploration function is $\tf : \SState \to \SState$,
which maps a set of states to their successors plus the initial state:
\begin{small}\begin{align*}
  \tf(\sstate) &= \setbuild{ \state' }{ \state \in \sstate \text{ and } \state \To \state' } \cup \set{ \state_0 }
  \text,
\end{align*}\end{small}%
Because the function is monotonic, there exists a fixed
point
\begin{small}\begin{align*}
  S = \bigsqcup_{n = 0}^\infty f^n(\emptyset)\text,
\end{align*}\end{small}%
which is the (possibly infinite) set of reachable states.

\subsection{Abstract semantics for CPS: $\boldsymbol k$-CFA}

The development of the abstract semantics parallels the construction
of the concrete semantics.
The abstract state-space is structurally similar to the concrete
semantics:
\begin{small}\begin{align*}
  \astate \in \aState &=
  \syn{Call} \times 
  \sa{BEnv} \times 
  \sa{Store} \times
  \sa{Time}
  \\
  \abenv \in \sa{BEnv} &= \syn{Var} \to \sa{Addr}
  \\
  \astore \in \sa{Store} &= \sa{Addr} \to \aDen
  \\
  \aden \in \aDen &= \Pow{\sa{Clo}}
  \\
  \aclo \in \sa{Clo} &= \syn{Lam} \times \sa{BEnv}
  \\
  \aaddr \in \sa{Addr} &\text{ is a \textbf{finite} set of addresses}
  \\
  \atm \in \sa{Time} &\text{ is a \textbf{finite} set of time-stamps}
  \text.
\end{align*}\end{small}%
There are three major distinctions with the concrete state-space: (1)
the set of time-stamps is finite; (2) the set of addresses is finite;
and (3) the store can return a \emph{set} of values.
We assume the natural partial order $(\wt)$ on this state-space
and its components, along with the associated meaning for least-upper
bound ($\join$).  
For example:
\begin{small}\begin{align*}
  \astore \join \astore' &= \lambda \aaddr .( \astore(\aaddr) \cup
  \astore'(\aaddr) )
  \text.
\end{align*}\end{small}%

A state-wise abstraction map $\absmap : \State \to \aState$ formally relates the concrete
state-space to the abstract state-space:
\begin{small}
\begin{gather*}
\begin{align*}
  \absmap(\call,\benv,\store,\tm) &=
  (\call, \absmap(\benv), \absmap(\store), \absmap(\tm))
  \\
  \absmap(\benv) &= \lambda v . \absmap(\benv(v))
  \\
  \absmap(\store) &= \lambda \aaddr . \!\!\! \bigjoin_{\absmap(\addr) = \aaddr} \!\!\! \absmap(\store(\addr))
  \\
  \absmap(\lam,\benv) &= \set{(\lam,\absmap(\benv))}
\end{align*}\\
\begin{align*}
  \absmap(\addr) &\text{ is fixed by }\aalloc
  &
  \absmap(\tm) &\text{ is fixed by }\atick\text.
\end{align*}\end{gather*}\end{small}%
We cannot choose an abstraction for addresses and time-stamps until we
have chosen the sets $\s{Time}$, $\sa{Time}$, $\s{Addr}$ and
$\sa{Addr}$.

The initial abstract state for a program $\call$ is the direct
abstraction of the initial concrete state:
\begin{small}\begin{align*}
  \astate_0 = \absmap(\state_0) = (\call,\bot,\bot,\absmap(\tm_0))
  \text.
\end{align*}\end{small}%
The abstract semantics has an expression evaluator:
\begin{small}
\begin{gather*}
  \aEval : \syn{Exp} \times \sa{BEnv} \times \sa{Store} \to \aDen
  \\
\begin{align*}
  \aEval(\vv,\abenv,\astore) &= \astore(\abenv(\vv))
  &
  \aEval(\lam,\abenv,\astore) &= \{ (\lam, \abenv) \}
  \text.
\end{align*}\end{gather*}\end{small}%

The abstract transition relation $(\aTo) \subseteq
\aState \times \aState$ mimics its concrete counterpart as well; when
$\call = \sembr{\appform{f}{e_1 \ldots e_n}^\lab}$:
\begin{gather*}
  (\call,\abenv,\astore,\atm) 
  \aTo
  (\call',\abenv'',\astore',\atm')\text{, where }
  \\
\begin{small}\begin{align*}
  (\lam,\abenv') &\in \aEval(f,\abenv,\astore)
  &
  \aden_i &= \aEval(e_i,\abenv,\astore)
  \\
  \lam &= \sembr{\lamform{v_1 \ldots v_n}{\call'}^{\lab'}}
  &
  \atm' &= \atick(\call,\atm)
  \\
  \aaddr_i &= \aalloc(v_i,\atm')
  &
  \abenv'' &= \abenv'[v_i \mapsto \aaddr_i]
  \\
  \astore' &= \astore \join [\aaddr_i \mapsto \aden_i]\text.
\end{align*}\end{small}\end{gather*}
Notable differences are the fact that this rule is non-deterministic
(it branches to every abstract closure to which the function $f$
evaluates), and that every abstract address could represent several
concrete addresses, which means that additions to the store must be
performed with a join operation $(\join)$ rather than an extension.
There are also external parameters for the abstract semantics
corresponding to the external parameters of the concrete semantics:
\begin{small}\begin{align*}
  \atick &: \syn{Call} \times \sa{Time} \to \sa{Time}
  \\
  \aalloc &: \syn{Var} \times \sa{Time} \to \sa{Addr}
\end{align*}\end{small}%
The $\atick$ function allocates an abstract time, which is allowed to
be an abstract time which has been allocated previously; the allocator
$\aalloc$ is similarly allowed to re-allocate previously-allocated addresses.

\subsection{Constraints from soundness}

The standard soundness theorem requires that the abstract semantics
simulate the concrete semantics; the key inductive step shows
simulation across a single transition:
\begin{theorem}
  If 
  \(
    \state \To \state' \text{ and } \absmap(\state) \wt \astate
    \text,
    \)
  then there must exist an abstract state $\astate'$ such that:
  \(
    \astate \To \astate' \text{ and } \absmap(\state') \wt \astate'
    \text.
    \)
\end{theorem}
The proof reduces to two lemmas which must be
proved for every choice of the sets $\s{Time}$, $\sa{Time}$,
$\s{Addr}$ and $\sa{Addr}$:
\begin{lemma}
\label{lemma:time-stamp-simulation}
$\text{If } 
  \absmap(\tm) \wt \atm
  \text{, then }
  \absmap(\tick(\call,\tm)) \wt \atick(\call,\atm)
  \text.$
\end{lemma}
\begin{lemma}
\label{lemma:allocator-simulation}
$\text{If } 
  \absmap(\tm) \wt \atm
  \text{, then }
  \absmap(\alloc(\vv,\tm)) \wt \aalloc(\vv,\atm)
  \text.$
\end{lemma}


\subsubsection{The $\boldsymbol k$-CFA solution}

\kCFA{} represents one solution to the Simulation
Lemmas~\ref{lemma:time-stamp-simulation} and
\ref{lemma:allocator-simulation}.
In \kCFA{}, a concrete time-stamp is the sequence of call
sites traversed since the start of the program; an abstract
time-stamp is the last $k$ call sites.
An address is a variable plus its binding time:
\begin{small}\begin{align*}
  \s{Time} &= \syn{Call}^*                    & \sa{Time} &= \syn{Call}^k
  \\
  \s{Addr} &= \syn{Var} \times \s{Time} \;\;\;\;\;\;\;\;\;\;\;
  & \sa{Addr} &= \syn{Var} \times \sa{Time}
  \text.
\end{align*}\end{small}%
In theory, \kCFA{} is able to distinguish up to $\abs{\syn{Call}}^k$
instances (variants) of each variable---one for each invocation
context.
Of course, in practice, each variable tends to be bound in only a
small fraction of all possible invocation contexts.
Under this allocation regime, the external parameters are easily fixed:
\begin{small}\begin{align*}
  \mathit{tick}(\call,\tm) & = \call : \tm   \;\;\;\;\;\;\;\;\;\;\;
  & \atick(\call,\atm) &= \mathit{first}_k(\call : \atm)
  \\
  \alloc(\vv,\tm) & = (\vv,\tm)             & \aalloc(\vv,\atm) &= (\vv,\atm)
  \text,
\end{align*}\end{small}%
which leaves only one possible choice for the abstraction maps:
\begin{small}\begin{align*}
  \absmap(\tm) &= \mathit{first}_k(\tm)
  &
  \absmap(\vv,\tm) &= (\vv,\absmap(\tm))
  \text.
\end{align*}\end{small}%
In technical terms, $\atick$ determines the context-sensitivity of the
analysis, and $\aalloc$ determines its polyvariance.

\subsection{Computing $\boldsymbol k$-CFA na\"ively}

\kCFA{} can be computed na\"ively by finding the set of reachable
states.
The ``system-space'' for this approach is a set of states:
\begin{small}\begin{align*}
  \asstate \in \aSState &= \Pow{\aState}
  \text.
\end{align*}\end{small}%
The transfer function for this system-space is $\atf : \aSState \to \aSState$:
\begin{small}\begin{align*}
  \atf(\asstate) &= \setbuildsm{ \astate' }{ \astate \in \asstate \text{ and } \astate \aTo \astate' } \cup \set{ \astate_0 }
  \text.
\end{align*}\end{small}%
The size of the state-space bounds the complexity of na\"ive
\kCFA{}:\footnote{Because $\aalloc(v,\tm) = (v,\tm)$, we could encode
  every binding environment with a map from variables to just times,
  so  that, effectively, $|\sa{BEnv}| = |\syn{Var} \parto \sa{Time}| =
  |\sa{Time}|^{|\syn{Var}|} = |\syn{Call}|^{k\times{|\syn{Var}|}}$.}
\begin{small}\begin{align*}
  |\syn{Call}|
  \times 
  \overbrace{|\syn{Call}|^{k\times{|\syn{Var}|}}}^{|\sa{BEnv}|}
  \times
  \overbrace{\left(2^{|\syn{Lam}|\times|\syn{Call}|^{k \times {|\syn{Var}|}}}\right)^{|\syn{Var}|\times |\syn{Call}|^k}}^{|\sa{Store}|}
  \times 
  \overbrace{|\syn{Call}|^k}^{|\sa{Time}|}
\end{align*}\end{small}%
Even for $k = 0$, this method is deeply exponential, rather than the
expected cubic time more commonly associated with 0CFA.

\subsection{Computing $\boldsymbol k$-CFA with a single-threaded store} 

Shivers's technique for making \kCFA{} more efficient uses one
store to represent all stores.
Any set of stores may be conservatively approximated by 
their least-upper-bound.
Under this approximation, the system-space needs only one
store:
\begin{small}\begin{align*}
  \aSState =
  \Pow{\syn{Call} \times \sa{BEnv} \times \sa{Time}} \times \sa{Store}
  \text.
\end{align*}\end{small}%
Over this system-space, the transfer function becomes:
\begin{small}\begin{align*}
  \atf(\hat C,\astore) &= (\hat C \cup \hat C', \astore')
  \\
  \hat S' &= \setbuild{ \astate' }{ \hat c \in \hat C \text{ and } (\hat c,\astore) \aTo \astate' }
  \\
  \hat C' &= \setbuild{ \hat c }{ (\hat c,\astore) \in \hat S' }
  \\
  \astore' &= \!\!\!\! \bigjoin_{ (\hat c,\astore) \in \hat S' } \!\!\!\! { \astore }
  \text.
\end{align*}\end{small}%
[This formulation of the transfer function assumes
that the store grows monotonically across transition, \ie, that
$(\ldots,\astore,\atm) \aTo (\ldots,\astore',\tm')$ implies $\astore
\wt \astore'$.]

To compute the complexity of this analysis, note the isomorphism in the system-space:
\begin{small}\begin{equation*}
  \aSState \cong
  \left(
    \syn{Call} \to \Pow{\sa{BEnv} \times \sa{Time}}\right)
  \times
  \left(
    \sa{Addr} \to \Pow{\sa{Clo}}
  \right)
  \text,
\end{equation*}\end{small}%
Because the function $\atf$ is monotonic, 
the height of the lattice $\aSState$:
\begin{small}\begin{align*}
  &
  \abs{\syn{Call}}
  \times
    \overbrace{
  \abs{\syn{Call}}^{k\times{\abs{\syn{Var}}}}
}^{\abs{\sa{BEnv}}}
  \times
  \overbrace{
  \abs{\syn{Call}}^k
  }^{\abs{\sa{Time}}}
  \\ 
  +\; &
  \overbrace{
  \abs{\syn{Var}} 
  \times 
  \abs{\syn{Call}}^k
  }^{\abs{\sa{Addr}}}
  \times
  \overbrace{
    \abs{\syn{Lam}}
    \times
    \abs{\syn{Call}}^{k \times {\abs{\syn{Var}}}}
  }^{\abs{\sa{Clo}}}
  \text,
\end{align*}\end{small}%
bounds the maximum number of times we may have to apply the
abstract transfer function. 
For $k = 0$, the height of the lattice is quadratic in the size of the
program (with the cost of applying the transfer function linear
in the size of the program).
For $k \geq 1$, however, the algorithm has a genuinely exponential
system-space.

\section{Shivers's $\boldsymbol k$-CFA{} for Java}
\label{sec:kcfa-java}

Having formulated a small-step \kCFA{} for CPS, it is straightforward
to formulate a small-step, abstract interpretive \kCFA{} for Java.
To simplify the presentation, we utilize Featherweight
Java~\cite{dvanhorn:Igarashi:TOPLAS:2001} in ``A-Normal'' form.
A-Normal Featherweight Java is identical to ordinary Featherweight
Java, except that arguments to a function call must be atomically
evaluable, as they are in A-Normal Form $\lambda$-calculus.
For example, the body {\tt return f.foo(b.bar());} becomes the sequence of statements {\tt B b1
  = b.bar(); F f1 = f.foo(b1); return f1;}.
This shift does not change the expressive power of the language or the
nature of the analysis, but it does simplify the semantics by
eliminating semantic expression contexts.
The following grammar describes A-Normal Featherweight Java; note the
(re-)introduction of statements:
\begin{small}\begin{align*}
  \syn{Class} &\produces \classform{\className}{\className'}{
    \overrightarrow{\className''\; \fieldName\; \ttsc}\; 
    \constDef\;
    \overrightarrow{\methodDef}
  }
  \\
  \constDef \in \syn{Konst} 
  &
  \produces
  \className\; \ttlp 
   \overrightarrow{
   \className\; \fieldName\;
  }
  \ttrp
  \ttlc
  \mathtt{super}\ttlp 
   \overrightarrow{\fieldName'}
  \ttrp\;
  \ttsc\;
   \overrightarrow{\mathtt{this} \ttfs \fieldName'' \mathrel{=} \fieldName''' \ttsc}
  \ttrc
  \\
  \methodDef \in \syn{Method} 
  &
  \produces
  \className\; \methodName\; \ttlp 
   \overrightarrow{
   \className\; \vv\;
  }
  \ttrp\;
  \ttlc\;
  \overrightarrow{\className\; \vv\; \ttsc}\;
  \vec{\stmt}
  \;\ttrc
  \\
  \stmt \in \syn{Stmt} &
  \produces
  \vv = 
  \expr
  \;
  \ttsc^\lab
  \opor
  \texttt{return}\; \vv\; \ttsc^\lab
  \\
  \expr \in \syn{Exp} &
  \produces
  \vv 
  \opor 
  \vv \ttfs \fieldName
  \opor
  \vv \ttfs \methodName \ttlp 
   \overrightarrow{\vv}
  \ttrp
  \opor
  \mathtt{new} \; \className\; \ttlp 
   \overrightarrow{\vv}
  \ttrp
  \opor
  \ttlp \className \ttrp \vv
\end{align*}\end{small}%
\begin{small}\begin{align*}
  \fieldName \in \syn{FieldName} &= \syn{Var}
  \\
  \className \in \syn{ClassName} &\text{ is a set of class names}
  \\
  \methodName \in \syn{MethodCall} &\text{ is a set of method invocation sites}
  \\
  \lab \in \syn{Lab} &\text{ is a set of labels}
\end{align*}\end{small}%
The set $\syn{Var}$ contains both variable and field names.
Every statement has a label.
%
%
The function $\ssucc : \syn{Lab} \parto \syn{Stmt}$ yields the
subsequent statement for a statement's label.

\subsection{Concrete semantics for Featherweight Java}

\begin{figure}
\begin{small}
\begin{small}\begin{align*}
  \state \in \State &= 
  \syn{Stmt} \times 
  \s{BEnv} \times 
  \s{Store} \times
  \s{KontPtr} \times
  \s{Time}  
  \\
  \benv \in \s{BEnv} &= \syn{Var} \parto \s{Addr} 
  \\
  \store \in \s{Store} &= \s{Addr} \parto \s{D}
  \\
  \den \in \Den &= \s{Val}
  \\
  \val \in \s{Val} &= \s{Obj} + \s{Kont}
  \\
  \obj \in \s{Obj} &= \syn{ClassName} \times \s{BEnv}
  \\
  \cont \in \s{Kont} &= 
   \syn{Var} \times \syn{Stmt} \times \s{BEnv} \times \s{KontPtr}
  \\
  \addr \in \s{Addr} &\text{ is a set of addresses}
  \\
  \contptr \in \s{KontPtr} &\subseteq \s{Addr}
  \\
  \tm \in \s{Time} &\text{ is a set of time-stamps}  
  \text.  
\end{align*}\end{small}%
\end{small}%
\caption{Concrete state-space for A-Normal Featherweight Java.}
\label{fig:java-concrete-state-space}
\end{figure}

Figure~\ref{fig:java-concrete-state-space} contains the concrete
state-space for the small-step Featherweight Java machine, and
Figure~\ref{fig:java-concrete-semantics} contains the concrete
semantics.\footnote{
  Note that the $(+)$ operation represents right-biased functional
  union, and that wherever a vector $\vec{x}$ is in scope, its
  components are implicitly in scope: $\vec{x} = \vect{x_0, \ldots,
    x_{\mathit{length(\vec{x})}}}$.  }
The state-space closely resembles the concrete state-space for CPS.
One difference is the need to explicitly allocate continuations (from
the set $\s{Kont}$) at a semantic level.
These same continuations exist in CPS, but they're hidden in plain
sight---the CPS transform converts semantic
continuations into syntactic continuations.

It is important to note the encoding of objects:
objects are a class plus a record of their fields, and
the record component is encoded as a binding environment that maps
field names to their addresses.
This encoding is congruent to \kCFA's encoding of closures, but it is
probably not the way one would encode the record component of an
object if starting from scratch.
The natural encoding would reduce an object to a class plus a single
base address, \ie, $\Obj = \syn{ClassName} \times \s{Addr}$, since
fields are accessible as offsets from the base address.
Then, given an object $(\className,\addr)$, the address of field $f$
would be $(f,\addr)$.
In fact, under our semantics, given an object $(\className,\benv)$, it
is effectively the case that $\benv(f) = (f,\addr)$.
We are choosing the functional representation of records to
maintain the closest possible correspondence with CPS. 
When investigating the complexity of \kCFA{} for Java, we will exploit
this observation: the fact that objects can be represented with just
a base address causes the collapse in complexity.

The concrete semantics are encoded as a small-step transition relation
$(\To) \subseteq \State \times \State$.
Each expression type gets a transition rule.
\emph{Object allocation creates a new binding environment
  $\benv'$, which shares no structure with the previous environment
  $\benv$; contrast this with CPS.}
These rules use the helper functions described in
Figure~\ref{fig:concrete-anfw-java-helper}.
The constructor-lookup function $\FetchRuctor$ yields the field names
and the constructor associated with a class name.
A constructor $\Ructor$ takes newly allocated addresses to use for
fields and a vector of arguments; it returns the change to the store
plus the record component of the object that results from running the
constructor.
The method-lookup function $\MethodLookup$ takes a method invocation point and an
object to determine which method is actually being called at that
point.

\begin{figure}
  \begin{small}\begin{align*}
    \FetchRuctor &: \syn{ClassName} \to (\syn{FieldName}^* \times \s{Ructor})
    \\
  \Ructor &\in 
  \s{Ructor} = 
  \overbrace{\s{Addr}^*}^{\text{fields}} 
  \times 
  \overbrace{\Den^*}^{\text{arguments}}
  \to
  (
  \overbrace{
    \s{Store}
  }^{\text{field values}}
  \times 
  \overbrace{
    \s{BEnv}
  }^{\text{record}}
  )
  \\
  \MethodLookup & : \s{D} \times \syn{MethodCall} \parto \syn{Method}
  \end{align*}\end{small}%

  \caption{Helper functions for the concrete semantics.}
  \label{fig:concrete-anfw-java-helper}
\end{figure}

\begin{figure}
\paragraph{Variable reference}
\begin{gather*}
  (
  \sembr{\vv = \vv' \ttsc^\lab},
  \benv,
  \store,
  \contptr,
  \tm
  )
  \To
  (
  \ssucc(\lab),
  \benv,
  \store',
  \contptr,
  \tm'
  )\text{, where }\\
\begin{small}\begin{align*}
  \tm' &= \tick(\lab, \tm) 
  &
  \store' &= \store[\benv(\vv) \mapsto \store(\benv(\vv'))]
  \text.
\end{align*}\end{small}%
\end{gather*}

\paragraph{Return}
\begin{gather*}
  (\sembr{\texttt{return}\; \vv\; \ttsc^\lab}, \benv, \store, \contptr, \tm)
  \To
  (\stmt, \benv', \store', \contptr', \tm')\text{, where }
  \\
\begin{small}\begin{align*}
  \tm' &= \tick(\lab, \tm)
  &
  (\vv',\stmt, \benv', \contptr') &= \store(\contptr)
  \\
  \den &= \store(\benv(\vv))
  &
  \store' &= \store[\benv'(\vv') \mapsto \den]
  \text.
\end{align*}\end{small}%
\end{gather*}

\paragraph{Field reference}
\begin{gather*}
  (
  \sembr{\vv = \vv'\ttfs \fieldName\; \ttsc^\lab},
  \benv,
  \store,
  \contptr,
  \tm
  )
  \To
  (
  \ssucc(\lab),
  \benv,
  \store',
  \contptr,
  \tm'
  )\text{, where }\\
\begin{aligned}
  \tm' & = \tick(\lab, \tm)
  &
  (\className, \benv') &= \store(\benv(\vv'))  
  &
  \store' &= \store[\benv(\vv) \mapsto \store(\benv'(\fieldName))]
  \text.
\end{aligned}
\end{gather*}

\paragraph{Method invocation}
\begin{gather*}
\begin{split}
  &(
  \sembr{\vv = \vv_0 \ttfs \methodName \ttlp
    \overrightarrow{\vv'}
  \ttrp \ttsc^\lab},
  \benv,
  \store,
  \contptr,
  \tm
  )
  \To
  (
  \stmt_0,
  \benv'',
  \store',
  \contptr',
  \tm'
  )\text{,}\\  
  &\text{where }
\end{split}
  \\
  \methodDef = 
  \sembr{\className\; \methodName\; \ttlp 
   \overrightarrow{
   \className\; \vv''\;
  }
  \ttrp\;
  \ttlc
  \overrightarrow{\className'\; \vv'''\; \ttsc}\;
  \vec{\stmt}
  \ttrc}
  = \MethodLookup(\den_0,\methodName)\\
\begin{small}\begin{align*}
  \den_0 &= \store(\benv(\vv_0))
  &
  \den_i &= \store(\benv(\vv'_i))
  \\
  \tm' &= \tick(\lab,\tm)
  &
  \cont &= (\vv,\ssucc(\lab), \benv, \contptr)
  \\
  \contptr' &= \alloc_\cont(\methodDef,\tm')
  &
  \addr'_i &= \alloc(\vv_i'',\tm')
  \\
  \addr''_j &= \alloc(\vv_j''',\tm')
  &
  \benv' &= [\sembr{\tt this} \mapsto \benv(\vv_0)]
  \\
  \benv'' &= \benv'[\vv_i'' \mapsto \addr_i', \vv'''_j \mapsto \addr_j'']
  &
  \store' &= \store [\contptr' \mapsto \cont, \addr'_i \mapsto \den_i]
  \text.
\end{align*}\end{small}%
\end{gather*}

\paragraph{Object allocation}
\begin{gather*}
\begin{split}
  &(
  \sembr{
    \vv = 
    {\tt new}\; \className\; \ttlp
    \overrightarrow{\vv'}
  \ttrp \ttsc^\lab},
  \benv,
  \store,
  \contptr,
  \tm
  )
  \To
  (
  \ssucc(\lab),
  \benv,
  \store',
  \contptr,
  \tm'
  )\text{,}\\
&\text{where }
\end{split}\\
\begin{small}\begin{align*}
  \tm' &= \tick(\lab,\tm)
  &
  \den_i &= \store(\benv(\vv_i'))
  \\
  (\vec{\fieldName},\Ructor) &=
   \FetchRuctor(\className)
  &
  \addr_i &= \alloc(\fieldName_i,\tm')
  \\
  (\Delta \store, \benv') &= 
   \Ructor(\vec{\addr}, \vec{\den})
  &
  \den' &= (\className, \benv')
  \\
  \store' &= \store + \Delta \store + [\benv(\vv) \mapsto \den']
  \text.
\end{align*}\end{small}%
\end{gather*}

\paragraph{Casting}
\begin{gather*}
  (
  \sembr{\vv = \ttlp C'\ttrp\; \vv'},
  \benv,
  \store,
  \contptr,
  \tm
  )
  \To
  (
  \ssucc(\lab),
  \benv,
  \store',
  \contptr,
  \tm'
  )\text{, where}\\
\begin{small}\begin{align*}
  \tm' &= \tick(\lab, \tm)
  &
  \store' &= \store[\benv(\vv) \mapsto \store(\benv(\vv'))]
  \text.
\end{align*}\end{small}%
\end{gather*}
\caption{Concrete semantics for A-Normal Featherweight Java.}
\label{fig:java-concrete-semantics}
\end{figure}

\subsection{Abstract semantics: $\boldsymbol k$-CFA for Featherweight Java}

\begin{figure}
\begin{small}
\begin{small}\begin{align*}
  \astate \in \aState &= 
  \syn{Stmt} \times 
  \sa{BEnv} \times 
  \sa{Store} \times
  \sa{KontPtr} \times
  \sa{Time}  
  \\
  \abenv \in \sa{BEnv} &= \syn{Var} \parto \sa{Addr} 
  \\
  \astore \in \sa{Store} &= \sa{Addr} \to \sa{D}
  \\
  \aden \in \aDen &= \Pow{\sa{Val}}
  \\
  \aval \in \sa{Val} &= \sa{Obj} + \sa{Kont}
  \\
  \aobj \in \sa{Obj} &= \syn{ClassName} \times \sa{BEnv}
  \\
  \acont \in \sa{Kont} &= 
   \syn{Var} \times \syn{Stmt} \times \sa{BEnv} \times \sa{KontPtr}
  \\
  \aaddr \in \sa{Addr} &\text{ is a \textbf{finite} set of addresses}
  \\
  \acontptr \in \sa{KontPtr} &\subseteq \sa{Addr}
  \\
  \atm \in \sa{Time} &\text{ is a \textbf{finite} set of time-stamps}  
  \text.  
\end{align*}\end{small}%
\end{small}%
\caption{Abstract state-space for A-Normal Featherweight Java.}
\label{fig:java-abstract-state-space}
\end{figure}

Figure~\ref{fig:java-abstract-state-space} contains the abstract
state-space for the small-step Featherweight Java machine, \ie, OO
$k$-CFA.
As was the case for CPS, the abstract semantics closely mirror the
concrete semantics.
We assume the natural partial order for the components of the abstract
state-space.

The abstract semantics are encoded as a small-step transition relation
$(\aTo) \subseteq \aState \times \aState$, shown in
Figure~\ref{fig:abstract-anfw-java}.
There is one abstract transition rule for each expression type, plus
an additional transition rule to account for return.
These rules make use of the helper functions described in
Figure~\ref{fig:abstract-anfw-java-helper}. 
The constructor-lookup function $\aFetchRuctor$ yields the field names
and the abstract constructor associated with a class name.
An abstract constructor $\aRuctor$ takes abstract addresses to use for
fields and a vector of arguments; it returns the ``change'' to the
store plus the record component of the object that results from
running the constructor.
The abstract method-lookup function $\aMethodLookup$ takes a method
invocation point and an object to determine which methods could be
called at that point.

\begin{figure}
  \begin{small}\begin{align*}
    \aFetchRuctor &: \syn{ClassName} \to (\syn{FieldName}^* \times \sa{Ructor})
    \\
  \aRuctor &\in 
  \sa{Ructor} = 
    \sa{Addr}^*
  \times 
    \aDen^*
    \to
  (
  \sa{Store}
  \times 
  \sa{BEnv}
  )
  \\
  \aMethodLookup & : \sa{D} \times \syn{MethodCall} \to \Pow{\syn{Method}}
  \end{align*}\end{small}%

  \caption{Helper functions for the abstract semantics.}
  \label{fig:abstract-anfw-java-helper}
\end{figure}

\begin{figure}
\paragraph{Variable reference}
\begin{gather*}
  (
  \sembr{\vv = \vv' \ttsc^\lab},
  \abenv,
  \astore,
  \acontptr,
  \atm
  )
  \aTo
  (
  \ssucc(\lab),
  \abenv,
  \astore',
  \acontptr,
  \atm'
  )\text{, where }\\
\begin{small}\begin{align*}
  \atm' & = \atick(\lab, \atm)
  &
  \astore' &= \astore \join [\abenv(\vv) \mapsto \astore(\abenv(\vv'))]
  \text.
\end{align*}\end{small}%
\end{gather*}

\paragraph{Return}
\begin{gather*}
  (\sembr{\texttt{return}\; \vv\; \ttsc^\lab}, \abenv, \astore, \acontptr, \atm)
  \aTo
  (\stmt, \abenv', \astore', \acontptr', \atm')\text{, where }
  \\
\begin{small}\begin{align*}
  \atm' &= \atick(\lab, \atm)
  &
  (\vv',\stmt, \abenv', \acontptr') &\in \astore(\acontptr)
  \\
  \aden &= \astore(\abenv(\vv))
  &
  \astore' &= \astore \join [\abenv'(\vv') \mapsto \aden]
  \text.
\end{align*}\end{small}%
\end{gather*}

\paragraph{Field reference}
\begin{gather*}
  (
  \sembr{\vv = \vv'\ttfs \fieldName\; \ttsc^\lab},
  \abenv,
  \astore,
  \acontptr,
  \atm
  )
  \aTo
  (
  \ssucc(\lab),
  \abenv,
  \astore',
  \acontptr,
  \atm'
  )\text{, where }
  \\
\begin{small}\begin{align*}
  \atm' & = \atick(\lab, \atm)
  &
  (\className,\abenv') &\in \astore(\abenv(\vv'))
  &
  \astore' &= \astore \join [\abenv(\vv) \mapsto \astore(\abenv'(\fieldName))]
  \text.
\end{align*}\end{small}%
\end{gather*}

\paragraph{Method invocation}
\begin{gather*}
\begin{split}
  &(
  \sembr{\vv = \vv_0 \ttfs \methodName \ttlp
    \overrightarrow{\vv'}
  \ttrp \ttsc^\lab},
  \abenv,
  \astore,
  \acontptr,
  \atm
  )
  \aTo
  (
  \stmt_0,
  \abenv'',
  \astore',
  \acontptr',
  \atm'
  )\text{,}\\
  &\text{where }
\end{split}\\
  \methodDef = 
  \sembr{\className\; \methodName\; \ttlp 
   \overrightarrow{
   \className\; \vv''\;
  }
  \ttrp\;
  \ttlc
  \overrightarrow{\className'\; \vv'''\; \ttsc}\;
  \vec{\stmt}
  \ttrc}
  \in \MethodLookup(\aden_0,\methodName)
  \\
\begin{small}\begin{align*}
  \aden_0 &= \astore(\abenv(\vv_0))
  &
  \aden_i &= \astore(\abenv(\vv'_i))
  \\
  \atm' &= \atick(\lab,\atm)
  &
  \acont &= (\vv,\ssucc(\lab), \abenv, \acontptr)
  \\
  \acontptr' &= \aalloc_\acont(\methodDef,\atm')
  &
  \aaddr'_i &= \aalloc(\vv_i'',\atm')
  \\
  \aaddr''_j &= \aalloc(\vv_j''',\atm')
  &
  \abenv' &= [\sembr{\tt this} \mapsto \abenv(\vv_0)]
  \\
  \abenv'' &= \abenv'[\vv_i'' \mapsto \aaddr_i', \vv'''_j \mapsto \aaddr_j'']
  &
  \astore' &= \astore \join [\acontptr' \mapsto \set{\acont}, \aaddr'_i \mapsto \aden_i]
  \text.
\end{align*}\end{small}%
\end{gather*}

\paragraph{Object allocation}
\begin{gather*}
\begin{split}
  &(
  \sembr{
    \vv = 
    {\tt new}\; \className\; \ttlp
    \overrightarrow{\vv'}
  \ttrp\ttsc^\lab},
  \abenv,
  \astore,
  \acontptr,
  \atm
  )
  \aTo
  (
  \ssucc(\lab),
  \abenv,
  \astore',
  \acontptr,
  \atm'
  )\text{,}\\
  &\text{where }
\end{split}\\
\begin{small}\begin{align*}
  \atm' &= \atick(\lab,\atm)
  &
  \aden_i &= \astore(\abenv(\vv_i'))
  \\
  (\vec{\fieldName},\aRuctor) &=
   \aFetchRuctor(\className)
  &
  \aaddr_i &= \aalloc(\fieldName_i,\atm')
  \\
  (\Delta \astore, \abenv') &= 
   \aRuctor(\vec{\aaddr}, \vec{\aden})
  &
  \aden' &= (\className, \abenv')
  \\
  \astore' &= \astore \join \Delta \astore \join [\abenv(\vv) \mapsto \aden']
  \text.
\end{align*}\end{small}%
\end{gather*}

\paragraph{Casting}
\begin{gather*}
  (
  \sembr{\vv = \ttlp C'\ttrp\; \vv'},
  \abenv,
  \astore,
  \acontptr,
  \atm
  )
  \aTo
  (
  \ssucc(\lab),
  \abenv,
  \astore',
  \acontptr,
  \atm'
  )
  \\
\begin{small}\begin{align*}
  \atm' &= \atick(\lab, \atm)
  &
  \astore' &= \astore \join [\abenv(\vv) \mapsto \astore(\abenv(\vv'))]
  \text.
\end{align*}\end{small}%
\end{gather*}
\caption{Abstract semantics for A-Normal Featherweight Java.}
\label{fig:abstract-anfw-java}
\end{figure}

\subsection{The $\boldsymbol k$-CFA solution}

As in the original \kCFA{} for CPS, we factored out time-stamp and
address allocation functions and even the structure of time-stamps and
addresses.
The equivalent to call sites in Java are statements.
So, a concrete time-stamp is the sequence of labels traversed since
the program began execution.
Addresses pair either a variable/field name or a method with a time.
Method names are allowed, so that continuations can have a binding
point for each method at each time.
(Were method names not allowed, then all procedures would return to
the same continuations in ``0''CFA.)
\begin{small}\begin{align*}
  \s{Time} &= \syn{Lab}^*                    & \sa{Time} &= \syn{Lab}^k
  \\
  \s{Addr} &= \syn{Offset} \times \s{Time} 
  & \sa{Addr} &= \syn{Offset} \times \sa{Time}
  \\
  \syn{Offset} &= \syn{Var} + \syn{Method}
  \text.
\end{align*}\end{small}%
The time-stamp function prepends the most recent label.
The variable/field-allocation function pairs the variable/field with
the current time, while the continuation-allocation function pairs the
method being invoked with the current time:
\begin{small}\begin{align*}
  \mathit{tick}(\lab,\tm) & = \lab : \tm   \;\;\;\;\;\;\;\;\;\;\;
  & \atick(\lab,\atm) &= \mathit{first}_k(\lab : \atm)
  \\
  \alloc(\vv,\tm) & = (\vv,\tm)             & \aalloc(\vv,\atm) &= (\vv,\atm)
  \\
  \alloc_{\cont}(\methodDef,\tm) &= (\methodDef,\tm) & \aalloc_{\acont}(\methodDef,\atm) &= (\methodDef,\atm)
  \text.
\end{align*}\end{small}%
%
%

\subsection{Computing $\boldsymbol k$-CFA for Featherweight Java}

When we apply the single-threaded store optimization for \kCFA{} over
Java, the state-space appears to be genuinely exponential for $k \geq
1$.
This is because the analysis affords more precision and control over
individual fields than is normally expected of a pointer analysis.
Under \kCFA, the address of every field is the field name paired with
the abstract time from its moment of allocation; the same is true of
every procedure parameter.
However, these fields are still stored within maps, and these maps are
the source of the apparent complexity explosion.

Fortunately, by inspecting the semantics, we see that every address
in the range of a binding environment shares the same time.
Thus, binding environments ($\sa{BEnv}$) may be replaced
directly by the time of allocation with no loss of precision.
In effect, $\sa{BEnv} \cong \sa{Time}$ for object-oriented programs.
Simplifying the semantics under this assumption leads to an abstract
system-space with a polynomial number of bits to (monotonically) flip for a fixed $k$:
\begin{small}\begin{align*}
    & |\syn{Stmt}| \cdot |\sa{Time}|^3 \cdot |\syn{Method}| 
    + |\syn{Method} + \syn{Var}| \cdot |\s{Time}| 
    \\
    \cdot\; & (|\syn{Class}| \cdot |\s{Time}| + |\syn{Var}| \cdot |\syn{Stmt}| \cdot |\s{Time}| \cdot |\syn{Method}| \cdot |\s{Time}|)    
\end{align*}\end{small}

By constructing Shivers's \kCFA{} for Java, and noting the subtle
difference between the semantics' handling of closures and objects, we
have exposed the root cause of the discrepancy in complexity.
In the next section, we profit from this observation by constructing a
semantics in which closures behave like objects, resulting in a
polynomial-time, context-sensitive hierarchy of CFAs for functional
programs.

\subsection{Variations}

The above form of \kCFA{} is not exactly what would be usually called
a \kCFA{} points-to analysis in OO languages. Specifically, OO
\kCFA{}s would typically not change the context for each statement but
only for method invocation statements. An OO \kCFA{} is a
call-site-sensitive points-to analysis: the only context maintained is
call-sites. That is, abstract time would not ``tick'', except in the
method invocation rule of
Figure~\ref{fig:abstract-anfw-java}. Furthermore, the caller's context
would be restored on a method return, instead of just advancing the
abstract time to its next step. (This choice is discussed extensively
in the next sections.) These variations, however, are orthogonal to
our main point: The algorithm is polynomial because of the
simultaneous closing of all fields of an object.

\section{$\boldsymbol m$-CFA: Context-sensitive CFA in PTIME}
\label{sec:mcfa}

\kCFA{} for object-oriented programs is polynomial-time because it
collapses the records inside objects into base addresses.
It is possible to re-engineer the semantics of the $\lambda$-calculus
so that we achieve a similar collapse with the environments inside
closures.
In fact, the re-engineering corresponds to a well-known compiler
optimization technique for functional languages: flat-environment
closures~\cite{dvanhorn:Cardelli1984Compiling,dvanhorn:Appel1991Compiling}.
In flat-environment closures, the values of all free variables are
copied directly into the new environment.
As a result, one needs to keep track of only the base address of the
environment: any free variable is accessed as an offset.

This flat-environment re-engineering leads to the desired
polynomiality, an outcome first noted in the universal framework of
\citet{mattmight:Jagannathan:1995:Unified} (here ``JW'' for brevity).
Some caution must be taken in the use of flat environments; 
if used in conjunction with Shivers's \kCFA{}-style
``last-$k$-call-sites'' contour-allocation strategy, flat environments
achieve weak context-sensitivity in practice
(Section~\ref{sec:related}).
Jagannathan and Weeks
suggest several contour abstractions for control-flow analyses,
including using the last $k$ call sites and the top $m$ frames of the
stack.
Section~\ref{sec:related} argues quantitatively and qualitatively that
the top-$m$-frames approach is the right abstraction for flat environments.
To distinguish this approach from other possible instantiations of the
JW framework, we term the
resulting hierarchy \nCFA{}.
%
Additionally, we note that it is important to specify \nCFA{}
explicitly, as we do below, since its form does not straightforwardly
follow from past results. Specifically,
Jagannathan and Weeks
do specify the abstract
domains necessary for a stack-based ``polynomial \kCFA'' but do not give an
explicit abstract semantics that would produce the results of their
examples. This is significant because simply adapting the
JW concrete semantics to the abstract domains would not produce
\nCFA{} (or any other reasonable static analysis). The analysis cannot
just ``pop'' stack frames when a finite prefix of the call-stack is
kept. For instance, when the current context abstraction consists of
call-sites ($f$, $g$), popping the last call-site will result in a
one-element stack. What our analysis needs to do instead (on a
function return) is \emph{restore} the abstract environment of the
current caller.

\subsection{A concrete semantics with flat closures}

In the new state-space, an environment is a base address:
\begin{small}\begin{align*}
  \state \in \State &=
  \syn{Call} \times 
  \s{Env} \times
  \s{Store}
  \\
  \store \in \s{Store} &= \s{Addr} \parto \s{D}
  \\
  \den \in \Den &= \s{Clo}
  \\
  \clo \in \s{Clo} &= \syn{Lam} \times \s{Env}
  \\
  \addr \in \s{Addr} &= \syn{Var} \times \s{Env}
  \\
  \env \in \s{Env} &\text{ is a set of base environment addresses}
  \text.
\end{align*}\end{small}%
%
%

%
The expression-evaluator $\Eval : \syn{Exp} \times \s{Env} \times
\s{Store} \parto \s{D}$ creates a closure over the current
environment:
\begin{small}\begin{align*}
  \Eval(v,\env,\store) &= \store(v,\env)
  &
  \Eval(\lam,\env,\store) &= (\lam,\env) 
  \text.
\end{align*}\end{small}%
%
%
There is only one transition rule;
when $\call = \sembr{\appform{f}{e_1 \ldots e_n}}$:
\begin{gather*}
  (\call,\env,\store) 
  \To
  (\call',\env'',\store')\text{, where }
  \\
\begin{small}\begin{align*}
  (\lam,\env') &= \Eval(f,\env,\store)
  &
  \den_i &= \Eval(e_i,\env,\store)
  \\
  \lam &= \sembr{\lamform{v_1 \ldots v_n}{\call'}}
  &
  \env'' &= \mathit{new}(\call,\env)
  \\
  \set{x_1,\ldots,x_m} &= \free(\lam)  
  &
  \addr_{v_i} &= (v_i,\env'')
  \\
  \addr_{x_j} &= (x_j,\env'')
  &
  \den'_j &= \store(x_j,\env')
  \\
  \store' &= \store
  [\addr_{v_i} \mapsto \den_i]
  [\addr_{x_j} \mapsto \den'_j]
  \text.
\end{align*}\end{small}\end{gather*}

\subsection{Abstract semantics: $\boldsymbol m$-CFA}
The abstract state-space is similar to the concrete:
\begin{small}\begin{align*}
  \astate \in \aState &=
  \syn{Call} \times 
  \sa{Env} \times
  \sa{Store}
  \\
  \astore \in \sa{Store} &= \sa{Addr} \parto \sa{D}
  \\
  \aden \in \aDen &= \Pow{\sa{Clo}}
  \\
  \aclo \in \sa{Clo} &= \syn{Lam} \times \sa{Env}
  \\
  \aaddr \in \sa{Addr} &= \syn{Var} \times \sa{Env}
  \\
  \aenv \in \sa{Env} &\text{ is a set of base environments addresses}
  \text.
\end{align*}\end{small}%
The abstract evaluator $\aEval : \syn{Exp} \times \sa{Env} \times \sa{Store} \to \sa{D}$ 
also mirrors the concrete semantics:
\begin{small}\begin{align*}
  \aEval(v,\aenv,\astore) &= \astore(\vv,\aenv)
  &
  \aEval(\lam,\aenv,\astore) &= \set{(\lam,\aenv)}
  \text.
\end{align*}\end{small}%
%
%
There is only one transition rule;
when $\call = \sembr{\appform{f}{e_1 \ldots e_n}}$:
\begin{gather*}
  (\call,\aenv,\astore)
  \To
  (\call',\aenv'',\astore')\text{, where }
  \\
\begin{small}\begin{align*}
  (\lam,\aenv') &\in \aEval(f,\aenv,\astore)
  &
  \aden_i &= \aEval(\expr_i,\aenv,\astore)
  \\
  \lam &= \sembr{\lamform{v_1 \ldots v_n}{\call'}}
  &
  \aaddr_{x_j} &= (x_j,\aenv'')
  \\
  \aenv'' &= \widehat{\mathit{new}}(\call,\aenv,\lam,\aenv')
  &
  \aaddr_{v_i} &= (v_i,\aenv'')
  \\
  \set{x_1,\ldots,x_m} &= \free(\lam)
  &
  \aden'_j &= \astore(x_j,\aenv')
  \end{align*}\end{small}\\  
\begin{small}
  \astore' = \astore
  \join
  [\aaddr_{v_i} \mapsto \aden_i]
  \join
  [\aaddr_{x_j} \mapsto \aden'_j]
  \text.
\end{small}%
\end{gather*}

\subsection{Context-sensitivity}
The parameter which must be fixed for \nCFA{} is the new 
environment allocator.
To construct the right kind of context-sensitive analysis, we will
work backward---from the abstract to the concrete.
We would like it to be the case that when a procedure is invoked,
bindings to its parameters are separated from other bindings based on
calling context.
In addition, we need it to be the case that procedures return to the
calling context in which they were invoked.
(Bear in mind that ``returning'' in CPS means calling the 
continuation argument.)
Directly allocating the last $k$ call sites, as in \kCFA, does not
achieve the desired effect, because variables get repeatedly rebound
during the evaluation of a procedure with each invocation of an
internal continuation.
This causes variables from separate invocations to merge once 
they are $k$ calls into in the procedure.
Counterintuitively, we solve this problem by allocating \emph{fewer}
abstract environments.
We want to allocate a new environment when a true procedure is invoked, and
we want to restore an old environment when a continuation is invoked.
As a result, \nCFA{} is sensitive to the top $m$ stack frames, whereas
\kCFA{} is sensitive to the last $k$ calls.\footnote{ Consider a
  program which calls $a$, calls $b$ and then returns from $b$.
  [$k=1$]CFA will consider the context to be the call to $b$, while [$m=1$]CFA
  will consider the context to be the call to $a$.}

In this case, environments will be a function of context, so we have
environments play the role of time-stamps in \kCFA:
\begin{small}\begin{align*}
  \sa{Env} = \sa{Call}^m
  \text,
\end{align*}\end{small}%
\nCFA{} assumes and exploits the well-known partitioning of the CPS
grammar from $\Delta$CFA \cite{dvanhorn:Might:2006:DeltaCFA} which
syntactically distinguishes ordinary procedures from continuations:
\begin{small}\begin{align*}
  \anew(\call,\aenv,\lam,\aenv') &= 
  \begin{cases}
    \mathit{first}_m(\call : \aenv)
    & \lam \text{ is a procedure }
    \\
    \aenv'
    & \lam \text{ is a continuation}
    \text. 
  \end{cases}
\end{align*}\end{small}%
From this it is clear that $[m=0]$CFA and $[k=0]$CFA are actually the
same context-insensitive analysis.

By setting $\s{Env} = \Nats \times \syn{Call}^*$, it is straightforward
to construct a concrete allocator that the abstract allocator simulates:
\begin{small}\begin{align*}
  \new(\call,(n,\vec{\call}),\lam,(n',\vec{\call}')) = 
  \;\;\;\;\;\;\;\;\;\;\;\;\;
  \\
  \begin{cases}
    (n+1,\call : \vec{\call})
    & \lam \text{ is a procedure}
    \\
    (n+1,\vec{\call}')
    & \lam \text{ is a continuation}
    \text.
  \end{cases}
\end{align*}\end{small}

\subsection{Computing $\boldsymbol m$-CFA}

Consider the single-threaded system-space for \nCFA:
\begin{small}\begin{align*}
  \aSState &= \Pow{\syn{Call} \times \sa{Env}} \times \sa{Store}
  \\
  & \cong \left(\syn{Call} \to \Pow{\sa{Env}}\right)
  \times \left(\syn{Addr} \to \Pow{\sa{Clo}}\right)
  \text.
\end{align*}\end{small}%

\begin{theorem}
  Computing \nCFA{} is complete for PTIME.
\end{theorem}
\begin{proof}
  Computing \nCFA{} is a monotonic ascent through a lattice whose
  height is polynomial in program size:
\begin{small}\begin{align*}
  \abs{\syn{Call}} \times \abs{\syn{Call}}^m
  \times
  \abs{\syn{Var}} \times \abs{\syn{Call}}^m
  \times
  \abs{\syn{Lam}} \times \abs{\syn{Call}}^m
  \text.
\end{align*}\end{small}%
Clearly, for any choice of $m \geq 0$, \nCFA{} is computable in
polynomial time.
Hardness follows from the fact that $[m=0]$CFA and $[k=0]$CFA are the
same analysis, which is known to be PTIME-hard
\cite{dvanhorn:VanHorn-Mairson:ICFP07}.
\end{proof}


\section{Comparisons to related analyses}
\label{sec:implementation}
\label{sec:related}

This work draws heavily on the Cousots' abstract
interpretation~\cite{mattmight:Cousot:1977:AI,mattmight:Cousot:1979:Galois} and upon Shivers's 
original formulation of $k$-CFA~\cite{mattmight:Shivers:1991:CFA}.
\nCFA{} (assuming suitable widening) can be viewed as an instance of
the universal framework of \citet{mattmight:Jagannathan:1995:Unified},
but for continuation-passing style.
%
%
%
If one naively uses the framework of
\citet{mattmight:Jagannathan:1995:Unified} with Shivers's \kCFA{}
contour-allocation strategy, the result is a polynomial CFA algorithm
that uses a ``last-$k$-call-sites'' context abstraction, unlike our
\nCFA{}, which uses a top-$m$-frames abstraction.
In the rest of this section, ``naive polynomial $k$-CFA'' refers to a
flat-environment CFA with a last-$k$-call-sites abstraction.

We will argue next, both qualitatively and quantitatively,
why the top-$m$-frame abstraction is better than the last-$k$-call
abstraction for the case of flat-environment CFAs.
%
%
%
The distinction between these policies is subtle yet
important.
Using the last $k$ call sites forces environments within a function's
scope to merge after the $k$th (direct or indirect) call made by a
function.
Any recursive function will appear to make at least $k$ calls during
an analysis, leaving only leaf procedures with boosted
context-sensitivity; since leaf procedures do not invoke higher-order
functions, the extra context-sensitivity offers no benefit to
control-flow analysis.

Consider, for example, the invocation of a simple 
function:
\begin{code}
(identity 3)\end{code}
If the definition of the {\tt identity} function is:
\begin{code}
(define (identity x) x)\end{code}
then both naive polynomial 1CFA and $[m=1]$CFA return the same flow analysis
as $[k=1]$CFA for the program:
\begin{code}
(id 3)
(id 4)\end{code}
That is, all agree the return value is {\tt 4}.
If, however, we add a seemingly innocuous function call to the body of the identity function:
\begin{code}
(define (identity x)
        (do-something)
        x)\end{code} 
then polynomial 1CFA would say that the program returns {\tt 3} or {\tt 4}, 
whereas $[m=1]$CFA and $[k=1]$CFA still agree that the return value is just {\tt 4}.

To understand why naive polynomial 1CFA degenerates into the behavior of 0CFA with
the addition of the function call to {\tt do-something}, consider
what the last $k=1$ call sites are at the return point {\tt x}.
Without the intervening call to {\tt (do-something)}, the last call
site at this point was {\tt (id 3)} in the first case, and {\tt
  (id 4)} in the second case.
Thus, polynomial 1CFA keeps the bindings to {\tt x} distinct.
\emph{With} the intervening call to {\tt (do-something)}, the last 
call site becomes {\tt (do-something)} in both cases, causing
the flow sets for {\tt x} to merge together.
If, however, we allocate the top $m$ stack frames for the
environment, then the intervening call to {\tt do-something} has no
effect, because the top of the stack at the return point {\tt x} is
still the call to {\tt (id 3)} or {\tt (id 4)}, which keeps the
bindings distinct.



Several papers have investigated polyvariant flow analyses with
polynomial complexity bounds in the setting of type-based analysis, as
compared with the abstract interpretation approach employed in this
paper.
\citet{dvanhorn:Mossin:97:FlowAnalysis} presents a flow analysis based
on polymorphic subtyping including polymorphic recursion for a
simply-typed (i.e. monomorphically typed) $\lambda$-calculus.
Mossin's algorithm operates in $O(n^8)$-time and both
\citet{dvanhorn:Rehof:POPL01} and \citet{dvanhorn:Gustavsson:PADO01}
developed alternative algorithms that operate in $O(n^3)$, where $n$
is the size of the explicitly typed program (and in the worst case,
types may be exponentially larger than the programs they annotate).
\nCFA{} does not impose typability assumptions and is polynomial in
the program size without type annotations.
As a consequence of the abstract interpretation approach taken in
\nCFA{}, unreachable parts of the program are never analyzed,
in contrast to most type based approaches.
Another difference concerns the space of abstract values: \nCFA{}
includes closure approximations, while polymorphic recursive flow
types relate program text and do not predict run-time environment
structure.

\subsection{Benchmark-driven comparisons}

We have implemented $k$-CFA, \nCFA{} and polynomial $k$-CFA for R5RS
Scheme (with support for some of R6RS).
Making a fair comparison of unrelated CFAs (e.g., \nCFA{} and
polynomial $k$-CFA) is not straightforward.
CFAs are not totally ordered by either speed or precision for all
programs.
In fact, even within the same program, two CFAs may each be locally
more precise at different points in the program.
That is, given the output of two CFAs, it might not always be possible
to say one is more precise than another.
To compare CFAs on an ``apples-to-apples'' basis requires careful
benchmark construction; we discuss the results on such benchmarks
below.

\subsubsection{Comparing speed with precision held constant}
The constructive content of Van Horn and Mairson's proof offers a way
to generate benchmarks that exercise the worst-case behavior of a
CFA---by constructing a program that forces the CFA to the top of the
lattice (because the most precise possible answer is the top).
Using this insight, we constructed a series of successively larger
``worst-case'' benchmarks and recorded how long it took each CFA to
reach the top of the lattice on a 2 Core, 2 GHz OS X machine:
\begin{center}
{
\begin{tabular}{|c|c|c|c|c|}
\hline
\textsf{Terms} & \textsf{$k=1$} & \textsf{$m=1$} & \textsf{poly.,$k$=1} & \textsf{$k$=0}
\\
\hline
\hline
69 & $\epsilon$ &  $\epsilon$ & $\epsilon$ &  $\epsilon$
\\
\hline
123 & $\epsilon$ & $\epsilon$ & $\epsilon$ & $\epsilon$
\\
\hline
231 & 46 s & $\epsilon$ & 2 s & $\epsilon$
\\
\hline
447 & $\infty$ & 3 s & 5 s & 2 s
\\
\hline
879 & $\infty$ & 48 s & 1 m 8 s & 15 s
\\
\hline
1743 & $\infty$ & 51 m & $\infty$ &  3 m 48 s
\\
\hline
\end{tabular}
}
\end{center}
$\epsilon$ indicates that the analysis returned in less than one
second; $\infty$ indicates the analysis took longer than one hour.

As can be seen, \nCFA{} is not just faster than \kCFA{}
but also consistently faster than naive polynomial \kCFA{}. The difference
in scalability between \nCFA{} and \kCFA{} is large and matches
the theoretical expectations well.
\emph{From these numbers we can infer that, in the worst case, the
  feasible range of context-sensitive analysis of functional programs
  has been increased by two-to-three orders of magnitude.}

\subsection{Comparing speed and precision}

On the following benchmarks, we measured both the run-time of the
analyses and the number of inlinings supported by the results.
We are using the number of inlinings supported as a crude but
immediately practical metric of the precision of the analysis.
\begin{center}
{
\small
\begin{tabular}{|c|cc|cc|cc|cc|}
\hline
\textsf{Prog}/&
  \multicolumn{2}{|c|}{\multirow{2}{*}{\textsf{$k=1$}}} & 
  \multicolumn{2}{|c|}{\multirow{2}{*}{\textsf{$m=1$}}} &
  \multicolumn{2}{|c|}{\multirow{2}{*}{\textsf{poly.,$k$=1}}} &
  \multicolumn{2}{|c|}{\multirow{2}{*}{\textsf{$k$=0}}} \\
\textsf{Terms} & & & & & & & &
\\
\hline
\hline
\textsf{eta}
  & \multirow{2}{*}{$\epsilon$} & \multirow{2}{*}{7} 
  & \multirow{2}{*}{$\epsilon$} & \multirow{2}{*}{7}
  & \multirow{2}{*}{$\epsilon$} & \multirow{2}{*}{3}
  & \multirow{2}{*}{$\epsilon$} & \multirow{2}{*}{3}
\\
\textsf{49} & & & & & & & & \\
\hline
\textsf{map}
  & \multirow{2}{*}{$\epsilon$} & \multirow{2}{*}{8}
  & \multirow{2}{*}{$\epsilon$} & \multirow{2}{*}{8}
  & \multirow{2}{*}{$\epsilon$} & \multirow{2}{*}{8}
  & \multirow{2}{*}{$\epsilon$} & \multirow{2}{*}{6}
\\
\textsf{157} & & & & & & & & \\
\hline
\textsf{sat}
  & \multirow{2}{*}{$\infty$}   & \multirow{2}{*}{-}
  & \multirow{2}{*}{$\epsilon$} &  \multirow{2}{*}{12} 
  & \multirow{2}{*}{1s}         &  \multirow{2}{*}{12}  
  & \multirow{2}{*}{$\epsilon$} & \multirow{2}{*}{12} 
\\
\textsf{223} & & & & & & & & \\
\hline
\textsf{regex}
  & \multirow{2}{*}{4s}  &  \multirow{2}{*}{25}
  & \multirow{2}{*}{3s}  &  \multirow{2}{*}{25}
  & \multirow{2}{*}{14s} &  \multirow{2}{*}{25}
  & \multirow{2}{*}{2s}  &  \multirow{2}{*}{25} 
\\
\textsf{1015} & & & & & & & & \\
\hline
\textsf{scm2java}
  & \multirow{2}{*}{5s} &  \multirow{2}{*}{86}
  & \multirow{2}{*}{3s} &  \multirow{2}{*}{86} 
  & \multirow{2}{*}{3s} &  \multirow{2}{*}{79}
  & \multirow{2}{*}{4s} &  \multirow{2}{*}{79}
\\
\textsf{2318} & & & & & & & & \\
\hline
\textsf{interp}
  & \multirow{2}{*}{5s} &  \multirow{2}{*}{123}
  & \multirow{2}{*}{4s} &  \multirow{2}{*}{123}
  & \multirow{2}{*}{9s} &  \multirow{2}{*}{123}
  & \multirow{2}{*}{5s} &  \multirow{2}{*}{123}
\\
\textsf{4289} & & & & & & & & \\
\hline
\textsf{scm2c}
  & \multirow{2}{*}{179s} &  \multirow{2}{*}{136}
  & \multirow{2}{*}{143s} &  \multirow{2}{*}{136}
  & \multirow{2}{*}{157s} &  \multirow{2}{*}{131}
  & \multirow{2}{*}{55s} &  \multirow{2}{*}{131}
\\
\textsf{6219} & & & & & & & & \\
\hline
\end{tabular}
}
\end{center}
The first two benchmarks test common functional idioms;
\textsf{sat} is a back-tracking SAT-solver;
\textsf{regex} is a regular expression matcher based on derivatives;
\textsf{scm2java} is a Scheme compiler that targets Java;
\textsf{interp} is a meta-circular Scheme interpreter;
\textsf{scm2c} is a Scheme compiler that targets C.

\emph{From these experiments, \nCFA{} appears to be as precise as \kCFA{} in
practice, but at a fraction of the cost. Compared to naive polynomial 1CFA, 
$[m=1]$CFA is always equally fast or faster and equally or
more precise.
%
%
These experiments also suggest that naive polynomial 1CFA is little better
than 0CFA in practice, and, in fact, it even incurs a higher running
time than \kCFA{} in some cases.}


\section{Conclusion}
\label{sec:conclusion}

Our investigation began with the $k$-CFA paradox: the apparent
contradiction between (1) Van Horn and Mairson's proof that $k$-CFA is
EXPTIME-complete for functional languages and (2) the existence of
provably polynomial-time implementations of $k$-CFA for
object-oriented languages.
We resolved the paradox by showing that the \emph{same} abstraction
manifests itself differently for functional and object-oriented
languages.
To do so, we faithfully reconstructed Shivers's $k$-CFA for
Featherweight Java, and then found that the mechanism used to
represent closures is degenerate for the semantics of Java.
This degeneracy is what causes the collapse into polynomial time.

With respect to standard practice in \kCFA{}, the bindings inside
closures may be introduced over time in several contexts, whereas the
fields inside an object are all allocated in the same context.
This allows objects to be represented as a class name plus the initial
context, whereas the environments inside closures must be a true map from
variables to binding contexts; this map causes the exponential blow-up
in complexity for functional $k$-CFA.
Armed with this insight, we constructed a concrete semantics for the
$\lambda$-calculus which uses flat environments---environments in
which free variables are accessed as offsets from a base pointer,
rather than through a chain of environments.
In fact, this environment policy corresponds to well-known
implementation techniques from the field of functional program
compilation.

Under abstraction, flat environments exhibit the same degeneracy as
objects, and the end result is a polynomial hierarchy of
context-sensitive control-flow analyses for functional languages.
Our empirical investigation found that coupling flat environments with
a last-$k$-call-sites policy for context-allocation offers negligible
benefits for precision compared with 0CFA.
To solve this problem, we constructed a polynomial CFA hierarchy which
allocates the top $m$ stack frames as its context: \nCFA{}.
According to our empirical evaluation, \nCFA{} matches $k$-CFA in
precision, but with faster performance.

\section{Future work}
\label{sec:future}

Our intent with this work was to build a bridge.
Now built, that bridge spans the long-separated worlds of
functional and object-oriented program analysis.
Having already profited from the first round-trip voyage, it is worth asking what else
may cross.

We believe that abstract garbage collection is a good candidate~\cite{mattmight:Might:2006:GammaCFA}.
At the moment, it has only been formulated for the functional world.
The abstract semantics for Featherweight Java make it possible to
adapt abstract garbage collection to the static analysis of
object-oriented programs.
We hypothesize that its benefits for speed and precision will carry
over.

Going in the other direction, the field of points-to analysis for
object-oriented languages has significant maturity and has developed a
more practical understanding for what parameters (e.g., context depth)
and approximations (e.g., maintaining different contexts for variables
vs. closures) tend to yield fruitful precision for client analyses.
There is a more intense emphasis on implementation (e.g., using binary
decision diagrams) and on evaluation, which should
be possible to translate to the functional setting.
Also, what the object-oriented community calls
\emph{shape analysis} appears to go by \emph{environment analysis} in
the functional community.
Peering across from the functional side of the bridge, shape analyses
seem far ahead of environment analyses in their sophistication.
We hypothesize that these shape-analytic techniques will be profitable
for environment analysis.




\paragraph{Acknowledgments:} %
We are grateful to Jan Midtgaard for comments and relevant references
to the literature. We thank Ond\v{r}ej Lhot\'{a}k for valuable
discussions. This work was funded by the National Science Foundation
under grant 0937060 to the Computing Research Association for the
CIFellow Project, which supports David Van Horn, as well as grants
CCF-0917774 and CCF-0934631.

\bibliographystyle{plainnat}
\bibliography{bibliography,ptranalysis}

\end{document}